\documentclass[11pt,letterpaper]{article}

\usepackage[utf8]{inputenc}
\usepackage[english]{babel}

\usepackage[dvipsnames,usenames]{xcolor}
\usepackage[colorlinks=true,pdfpagemode=UseNone,urlcolor=RoyalBlue,linkcolor=RoyalBlue,citecolor=OliveGreen,pdfstartview=FitH]{hyperref}

\usepackage{tikz}
\usetikzlibrary{shapes.geometric}

\usepackage{subcaption}

\usepackage[margin=1in]{geometry}
\usepackage{amsfonts,amsmath,amsthm,thmtools}
\usepackage{nicefrac}

\declaretheorem{theorem}
\declaretheorem[sibling=theorem]{lemma}

\declaretheorem[sibling=theorem]{fact}

\usepackage{todonotes}

\usepackage{booktabs}
\usepackage{caption}
\usepackage[numbers, sort]{natbib}

\usepackage{tcolorbox}

\usepackage{enumitem}

\newcommand{\E}{\mathbf{E}}

\newcommand{\Sref}[1]{\hyperref[#1]{\S\ref*{#1}}}

\title{Edge-weighted Online Stochastic Matching: Beating $1-\frac1e$}

\author{
    Shuyi Yan
    \thanks{Department of Computer Science, University of Copenhagen. Supported by VILLUM Investigator Grant 16582, Basic Algorithms Research Copenhagen (BARC). Most of this work was done when the author was affiliated with Tsinghua University.}
}

\begin{document}
    
\begin{titlepage}
    \thispagestyle{empty}
    \maketitle
    \begin{abstract}
        \thispagestyle{empty}
        We study the edge-weighted online stochastic matching problem. Since \citet*{feldman2009online} introduced the online stochastic matching problem and proposed the $(1-\frac1e)$-competitive Suggested Matching algorithm, there has been no improvement in the edge-weighted setting. In this paper, we introduce the first algorithm beating the $1-\frac1e$ barrier in this setting, achieving a competitive ratio of $0.645$. Under the LP proposed by \citet*{jaillet2014online}, we design an algorithmic preprocessing, dividing all edges into two classes. Then we use different matching strategies to improve the performance on edges in one class in the early stage and on edges in another class in the late stage, while keeping the matching events of different edges highly independent. By balancing them, we finally guarantee the matched probability of every single edge.
    \end{abstract}
\end{titlepage}

\section{Introduction}

Since the online bipartite matching problem was introduced by \citet*{karp1990optimal}, it has been playing an important role in the field of online algorithms. One of the most important applications of online matching is online advertising. When a user searches on a search engine, appropriate advertisements should be selected to show. In this problem, advertisers are modeled as \emph{offline vertices} which are known to the algorithm at the beginning, and impressions (user searches) are modeled as \emph{online vertices} which arrive one by one. The edges between them indicate whether (or how much, when the edges have weights) advertisers are interested in impressions. When an online vertex arrives, its incident edges are revealed and the online algorithm should immediately decide how to match it.

\citet{karp1990optimal} proposed the \emph{worst case model}, which assumes the algorithm has no information about future online vertices, and measures its performance in the worst instance. Specifically, they considered the \emph{competitive ratio} of the algorithm, which is the ratio of the expected size (or expected total weight) of the matching by the algorithm to that of the offline optimal matching. They introduced the Ranking algorithm, achieving a competitive ratio of $1-\frac1e \simeq 0.632$ for unweighted matching, and proved its optimality.

To beat the $1-\frac1e$ bound, some new models are introduced, including the \emph{online stochastic matching} model proposed by \citet*{feldman2009online}. In online stochastic matching, we assume the algorithm has some prior knowledge about online vertices. In particular, a bipartite graph of offline vertices and \emph{online vertex types} is given at the beginning, along with a distribution over the online vertex types. When arriving, each online vertex independently samples its type from that distribution.
They introduced the ($1-\frac1e$)-competitive \emph{Suggested Matching}\footnote{They stated Suggested Matching and its related LP in the unweighted setting originally, but it's straightforward to generalize it to the edge-weighted setting (see Section \ref{sec:preliminaries} for more details).} algorithm as a benchmark, and proposed the first algorithm beating the $1-\frac1e$ bound in the \emph{unweighted} setting, under the further assumption of \emph{integral arrival rates}\footnote{The expected number of online vertices of each online vertex type is an integer.}, which were removed in some later works \cite{manshadi2012online, jaillet2014online, huang2021online, huang2022power}.
To generalize the results to weighted matching, \citet*{huang2021online} considered that the online vertices follow a \emph{Poisson process} to simplify the analysis. They proved that it is asymptotically equivalent\footnote{They are asymptotically equivalent as the number of online vertices goes to infinity, regardless of the type distribution, etc. Prior works in the original model also assumed a large number of online vertices.} to the original model where the number of online vertices is fixed.
The results are then generalized to \emph{vertex-weighted} matching \cite{huang2021online, huang2022power} and edge-weighted matching \emph{with free disposal}\footnote{An offline vertex can be matched multiple times, but only the heaviest edge remains in the end.} \cite{huang2022power}. On the other hand, under the integral assumption, the $1-\frac1e$ ratio was also surpassed in the edge-weighted setting \cite{haeupler2011online, brubach2020online}.

However, these algorithms only work for the special cases of edge-weighted matching, as they utilize some special properties of these models, which do not hold in the general edge-weighted setting. For example, in the vertex-weighted setting, it suffices to lowerbound the matched probability of each vertex to achieve a good competitive ratio. For the edge-weighted matching with free disposal, since matching an edge will not block the incident offline vertex and prevent it from matching with a heavier neighbor in the future, it's always a positive event to match an edge. It still suffices to lowerbound the expectation of the heaviest matched edge incident to each offline vertex. For the edge-weighted matching under the integral assumption, since each arrival rate is at least $1$, there is a nonnegligible probability that the same online vertex arrives twice (or more). So the offline optimum needs to match an online vertex to (at least) two offline neighbors with some probability to get a good performance, which means the online algorithm naturally has (at least) two matching candidates.

In the general edge-weighted setting, without the help of these useful properties, there were no positive results beyond Suggested Matching and its $1-\frac1e$ benchmark. Moreover, an upper bound of $0.703$ \cite{huang2022power} indicates that this setting is essentially harder than those special cases, as previous works achieved competitive ratios of $0.716$ \cite{huang2022power}, $0.706$ \cite{huang2022power} and $0.705$ \cite{brubach2020online} for vertex-weighted matching, edge-weighted matching with free disposal and edge-weighted matching under the integral assumption, respectively.

\subsection{Our Contributions and New Techniques}

In this paper, we propose a (polynomial-time) $0.645$-competitive algorithm for the edge-weighted online stochastic matching, which is the first algorithm beating the $1-\frac1e$ bound in this setting. Following the same framework of previous papers, we introduce our algorithm under Poisson arrivals, and we also solve an LP at the beginning to bound the offline optimum and guide our online algorithm. Below we briefly introduce our new techniques.

\paragraph{Preprocessing Under the Jaillet-Lu LP}

We use the LP proposed by \citet*{jaillet2014online}. Under this LP, we can preprocess the solution to make it satisfy some stronger constraints. We show that under the Jaillet-Lu LP,  without loss of generality, an online algorithm can assume that there are only two classes of online vertex types: A first-class online vertex type is matched with only one offline vertex (in the LP solution, which is a fractional matching) and a second-class one is matched with two offline vertices half to half.

\paragraph{Multistage Suggested Matching}

We take the Suggested Matching algorithm proposed by \citet{feldman2009online} as the starting point. After dividing edges into two classes, we optimize them separately. In Suggested Matching, the match strategy of first-class edges itself cannot be further improved, while that of second-class edges can. So on the one hand, we improve the strategy of second-class edges in the late stage to avoid affecting first-class edges too much, and on the other hand, we don't match second-class edges in a short period at the beginning to increase the match probabilities of first-class edges. Balancing the durations of these two operations, our algorithm achieves a competitive ratio over $1-\frac1e$ on every single edge.

\paragraph{Independent Analysis for Every Single Edge}

Previous analyses for vertex-weighted matching \cite{huang2021online,huang2022power} upperbound the unmatched probability of each offline vertex to calculate the competitive ratio. However, for edge-weighted matching, this cannot reflect the matched probability of each edge. On the other hand, the analysis for edge-weighted matching with free disposal \cite{huang2022power} pays attention to the overall matching progress, but it cannot handle the loss when a light edge occupies an offline vertex and then blocks a heavy edge. In contrast, our algorithm independently lowerbounds the unmatched probability of each offline vertex at any time, so that we can independently calculate an edge's matched probability given its match rate.

\subsection{Related Work}

In the worst case model, \citet{karp1990optimal} introduced the $(1-\frac1e)$-competitive Ranking algorithm, which was then generalized to the vertex-weighted setting by \citet{aggarwal2011online}, with the same ratio. For edge-weighted matching with free disposal, \citet{fahrbach2020edge} proposed a $0.5086$-competitive algorithm, then \citet{gao2022improved} and \citet{blanc2022multiway} independently improved the ratio to $0.519$ and $0.5239$ respectively.

In the random order model, where we assume the online vertices arrive in a random order after the adversary decides the graph, \citet{goel2008online} proved that the Greedy algorithm is $(1-\frac1e)$-competitive, and gave an upper bound of $\frac56$. \citet{mahdian2011online} improved the ratio to $0.696$. \citet{huang2019online} generalized it to vertex-weighted matching with a competitive ratio of $0.653$, and \citet*{jin2021improved} improved it to $0.662$.

In the online stochastic matching model, \citet{feldman2009online} proposed a $0.67$-competitive algorithm for unweighted matching under the integral assumption, and an upper bound of $0.989$. A sequence of works improved the competitive ratio to $0.716$ for vertex-weighted matching without the assumption \cite{bahmani2010improved,manshadi2012online,jaillet2014online,huang2021online,huang2022power,tang2022fractional}, and the upper bound was also improved to $0.823$ \cite{bahmani2010improved,manshadi2012online}.
\citet{huang2022power} proposed a $0.706$-competitive algorithm for edge-weighted matching with free disposal and proved an upper bound of $0.703$ for edge-weighted matching without free disposal. \citet{tang2022fractional} first considered the non i.i.d. arrivals and achieved a competitive ratio of $0.666$.

On the other hand, under the integral assumption, the competitive ratio for unweighted matching was improved to $0.7299$ \cite{manshadi2012online,jaillet2014online,brubach2020online}. \citet{jaillet2014online} achieved a competitive ratio of $0.725$ in the vertex-weighted setting. \citet{haeupler2011online} achieved a competitive ratio of $0.667$ in the edge-weighted setting, which was then improved to $0.705$ by \citet{brubach2020online}.

\section{Preliminaries}
\label{sec:preliminaries}

\paragraph{Edge-weighted Online Stochastic Matching}

Consider a set $I$ of online vertex types and a set $J$ of offline vertices. Let $J_i\subseteq J$ denote the set of offline vertices adjacent to online type $i\in I$. Let $E=\{(i,j):i\in I,j\in J_i\}$ denote the set of edges, where each edge $(i,j)$ has a non-negative weight $w_{ij}$. Each online vertex type $i\in I$ has an arrival rate $\lambda_i$, which is the expected number of online vertices of this type. $\Lambda=\sum_{i\in I}\lambda_i$ online vertices arrive one by one. Each online vertex draws its type $i$ with probability $\frac{\lambda_i}{\Lambda}$ independently. The objective is to maximize the expected total weight of all matched edges.

\paragraph{Poisson Arrivals}

Like previous works, we do not assume that the number of online vertices is fixed. Instead, we assume online vertices arrive following a Poisson process with an arrival rate $\Lambda$ in the time horizon $0 \le t \le 1$. Equivalently, the online vertices of each type $i$ arrive independently following a Poisson process with arrival rate $\lambda_i$. As we mentioned before, \citet{huang2021online} have proved it's asymptotically equivalent to the original stochastic model. We remark that our analysis would also work in the original model if we assume a large number of online vertices, which is also assumed in prior works in the original model.

\paragraph{Online Algorithms}

When an online vertex arrives, an online algorithm should immediately and irrevocably match it to an unmatched offline neighbor, or discard it. Define the competitive ratio of an online algorithm to be the infimum (over all possible instances) of the ratio of its expected objective to the expected objective of the offline optimal matching.

\paragraph{Suggested Matching}

Here we briefly restate the Suggested Matching algorithm proposed by \citet{feldman2009online}, in the edge-weighted setting. We first calculate the optimal solution of the following LP, which is a fractional matching upperbounding the offline optimum.
\begin{alignat}{2}
    \mbox{maximize}\quad & \sum_{(i,j)\in E} w_{ij} x_{ij} & {} & \nonumber \\
    \mbox{subject to}\quad & \sum_{j \in J} x_{ij}\leq\lambda_i &\quad& \forall i\in I \nonumber \\
    & \sum_{i\in I} x_{ij} \leq 1 &\quad& \forall j\in J \nonumber \\
    & x_{ij} \geq 0 &\quad& \forall(i,j)\in E \nonumber
\end{alignat}

Then, when an online vertex of type $i$ arrives, we match it to each unmatched adjacent offline vertex $j$ with probability $\frac{x_{ij}}{\lambda_i}$. In the end, the matched probability of each edge $(i,j)$ will be at least $(1-\frac{1}{e})x_{ij}$.

\paragraph{Jaillet-Lu LP}

We will use the LP proposed by \citet{jaillet2014online} instead of the above LP. We restate the Jaillet-Lu LP as follows:
\begin{alignat}{2}
    \mbox{maximize}\quad & \sum_{(i,j)\in E} w_{ij} x_{ij} & {} & \nonumber \\
    \mbox{subject to}\quad & \sum_{j \in J} x_{ij}\leq\lambda_i &\quad& \forall i\in I \nonumber \\
    & \sum_{i\in I} x_{ij} \leq 1 &\quad& \forall j\in J \nonumber \\
    & \sum_{i\in I} (2x_{ij}-\lambda_i)^+ \leq 1-\ln2 &\quad& \forall j\in J \label{eqn:yj-0} \\
    & x_{ij} \geq 0 &\quad& \forall(i,j)\in E \nonumber
\end{alignat}

The only difference between these two LPs is Constraint \eqref{eqn:yj-0}, which limits the situation that an online vertex type puts most of its arrival rate on one incident edge\footnote{Intuitively, it's because when the arrival rate of an online vertex type (or a set of types) becomes larger, the proportion of its arrival rate it can put on one edge becomes smaller.}. Although this LP was initially proposed for the original online stochastic matching model, \citet{huang2021online} have shown that it also bounds the offline optimum in the Poisson arrival model.
For convenience, we artificially define $x_{ij}=0$ for any $(i,j)\notin E$. We further define $x_i=\sum_{j\in J}x_{ij}$ and $x_j=\sum_{i\in I}x_{ij}$. We will also call $x_{ij}$ the \emph{flow} of edge $(i,j)$ or the \emph{flow} from $i$ to $j$.

\section{Preprocessing Under the Jaillet-Lu LP}

When we use the Jaillet-Lu LP to bound the offline optimum, we can preprocess the solution to make it satisfy some new constraints, while the objective remains unchanged and all original constraints still hold. This on the one hand helps in designing our algorithm, and on the other hand characterizes the worst instances an online algorithm needs to consider.

For each online vertex type $i$ such that $x_i<\lambda_i$, we add $m=\max\{\lceil\lambda_i-x_i\rceil,2\}$ dummy\footnote{For a dummy vertex, the weights of its incident edges are $0$, and its initial flow is $0$.} offline vertices connecting to $i$, each with a flow $\frac{\lambda_i-x_i}{m}$. The new dummy vertices satisfy Constraint \eqref{eqn:yj-0} since $m\ge 2$. It's straightforward to check that other constraints of the Jaillet-Lu LP still hold and we have:
\begin{equation}
\label{eqn:xi}
	\forall i\in I, \quad x_i=\lambda_i.
\end{equation}

Then we add two more dummy offline vertices, and one dummy online vertex type with a flow $1-x_j$ to every offline vertex $j$ (including the two dummy offline vertices), which has an appropriate arrival rate satisfying Equation \eqref{eqn:xi}. Since the arrival rate of this new vertex is at least $2$, it won't contribute to Constraint \eqref{eqn:yj-0} for any offline vertex. After that, all constraints of the Jaillet-Lu LP and Equation \eqref{eqn:xi} still hold, and we have:
\begin{equation}
\label{eqn:xj}
	\forall j\in J, \quad x_j=1.
\end{equation}

The final step is more complicated. We want to make the solution satisfy:
\begin{equation}
\label{eqn:edge}
	\forall (i,j)\in E, \quad x_{ij}\in\Big\{0,\frac12\lambda_i,\lambda_i\Big\},
\end{equation}
which means there are only two classes of online vertex types: one with an adjacent offline vertex such that $x_{ij}=\lambda_i$, and the other with two adjacent offline vertices such that $x_{ij_1}=x_{ij_2}=\frac12\lambda_i$. We will achieve this by finding a way to split each online vertex type $i$ to several types $i_1,\dots,i_k$, such that $\sum_{u=1}^k x_{i_u j}=x_{ij}$ for any offline vertex $j$. By assigning $\lambda_{i_u}=\sum_{j\in J}x_{i_u j}$, any split scheme trivially satisfies Equations \eqref{eqn:xi} and \eqref{eqn:xj}. So it suffices to find out how to pair up the edges to satisfy Equation \eqref{eqn:edge} while maintaining Constraint \eqref{eqn:yj-0}.
For example, an online vertex type with $2k$ equally weighted neighbors can be split into $k$ types, each with $2$ of the neighbors.

Informally speaking, any method that avoids pairing an edge with itself as much as possible will work. However, to give a formal proof, here we propose a specific split scheme.

\begin{lemma}
    For every online vertex type $i$, there is a split scheme satisfying Equation \eqref{eqn:edge} and Constraint \eqref{eqn:yj-0}.
\end{lemma}

\begin{proof}
    Consider an online vertex type $i$ with $k$ adjacent offline vertices $j_1$, \dots, $j_k$. Define a function $f_i:[0,\lambda_i)\to J_i$ such that $f_i(\theta)=j_u$ if and only if $\sum_{v=1}^{u-1}x_{ij_v}\le\theta<\sum_{v=1}^{u}x_{ij_v}$, i.e. each edge $(i,j_u)$ corresponds to an interval of length $x_{ij_u}$. Then we will pair $f_i(\theta)$ with $f_i(\theta+\frac{\lambda_i}{2})$.
    
    The interval $[0,\frac{\lambda_i}{2})$ can be divided into at most $2k$ subintervals, in which $f_i(\theta)$ and $f_i(\theta+\frac{\lambda_i}{2})$ are both invariant. For each subinterval $[l,r)$, create an online vertex type with a flow $r-l$ to each of $f_i(l)$ and $f_i(l+\frac{\lambda_i}{2})$. If $f_i(l)=f_i(l+\frac{\lambda_i}{2})$, they collapse to one edge with flow $2(r-l)$, and only in this case it contributes to Constraint \eqref{eqn:yj-0}.
    
    It's easy to see these online vertex types sum up to the original type $i$ and satisfy Equation \eqref{eqn:edge}. For each edge $(i,j_u)$ with $x_{ij_u}\le \frac{\lambda_i}{2}$, obviously there is no $\theta\in[0,\frac{\lambda_i}{2})$ such that $f_i(\theta)=f_i(\theta+\frac{\lambda_i}{2})=j_u$. If there is an edge $(i,j_u)$ with $x_{ij_u}> \frac{\lambda_i}{2}$, there will be only one subinterval $[l,r)$ such that $f_i(\theta)=f_i(\theta+\frac{\lambda_i}{2})=j_u$ for $\theta\in[l,r)$ with $r-l=x_{ij_u}-\frac{\lambda_i}{2}$. So Constraint \eqref{eqn:yj-0} still holds.
\end{proof}

We remark that, for an online algorithm, the split can be simulated as follows. When an online vertex of (original) type $i$ arrives, we change its type to $i_u$ with probability $\frac{\lambda_{i_u}}{\lambda_i}$.

In summary, for any solution of the Jaillet-Lu LP, we can preprocess it so that the constraints other than \eqref{eqn:yj-0} can be replaced by their tighter versions, Equations \eqref{eqn:xi}, \eqref{eqn:xj} and \eqref{eqn:edge}, so that the solution will have a more discrete structure. We summarize the conclusion in the following theorem.

\begin{theorem}
\label{thm:sol}
    For any solution of the Jaillet-Lu LP, we can transform it into an equivalent fractional matching (on another equivalent instance) with the same total weight, which satisfies Constraint \eqref{eqn:yj-0} and Equations \eqref{eqn:xi}, \eqref{eqn:xj} and \eqref{eqn:edge}.
\end{theorem}

We remark that the preprocessing can be done in polynomial time and it is independent of the online algorithm.

\section{Multistage Suggested Matching}

\subsection{Algorithm}

By Theorem \ref{thm:sol}, we have a fractional matching satisfying Constraint \eqref{eqn:yj-0} and Equations \eqref{eqn:xi}, \eqref{eqn:xj} and \eqref{eqn:edge}. Naturally, we ignore the edges with $x_{ij}=0$, so we say $i$ and $j$ are \emph{neighbors} if and only if $x_{ij}>0$. We call edge $(i,j)$ a \emph{first-class edge} if $x_{ij}=\lambda_i$, or a \emph{second-class edge} if $x_{ij}=\lambda_i/2$. If an online vertex/type has an incident first-class (respectively second-class) edge, we call it a \emph{first-class} (respectively \emph{second-class}) \emph{online vertex/type}. Let $I_1$ (respectively $I_2$) denote the set of first-class (respectively second-class) online vertex types. Define $y_j=\sum_{i\in I_1}x_{ij}$ to be the total flow of first-class edges incident to offline vertex $j$, then we can rewrite Constraint \eqref{eqn:yj-0} as:
\begin{equation}
\label{eqn:yj}
	\forall j\in J, \quad y_j\leq 1-\ln2.
\end{equation}

In the Suggested Matching algorithm, we match every first-class vertex to its only neighbor with probability $1$, while for a second-class vertex, we only match it to a neighbor with probability $\frac12$ even if the other neighbor has been matched. We try to increase the match rate in this case. However, if we do it with no limitations, such as raising the probability to $1$ when the other neighbor is matched, the matched probabilities of different offline vertices and edges will become highly correlated. A lighter edge may occupy its incident offline vertex with a higher probability so that a heavier edge will be more likely unmatched. In particular, it will hinder us from establishing a lower bound of the unmatched probability of an offline vertex at some point and then independently calculating the matched probability of each edge. Thus our algorithm only observes the status of offline vertices once, at some time point $t_1$. Only when a neighbor of some second-class vertex is matched at that time, we increase the match rate of the other neighbor in the remaining time.

The above method improves the performance on second-class edges, but the matched probabilities of first-class edges are affected. Since the match strategy of first-class vertices has no room to improve, we need to find another way. We sacrifice the match probabilities of second-class edges in early time to increase the unmatched probabilities of offline vertices. Before some time point $t_0$, our algorithm directly discards all arriving second-class vertices. Since the proportion of first-class edges is low (Eqn. \eqref{eqn:yj}), their matched probability will be significantly improved. On the contrary, the previous method slightly affects first-class edges since it only happens in the late stage. Combining these two methods, our algorithm makes the competitive ratio over every single edge exceed $1-\frac1e$.

Therefore, our algorithm can be divided into three stages by time, and in each stage, our match strategy is similar to that in Suggested Matching, so we call it the Multistage Suggested Matching algorithm.

\begin{tcolorbox}[beforeafter skip=10pt]
    \textbf{Multistage Suggested Matching}\\[1ex]
    \emph{Input at the beginning:}
    \begin{itemize}[itemsep=0pt, topsep=4pt]
        \item Online vertex types $I$, offline vertices $J$, edges $E$;
        \item Arrival rates $(\lambda_i)_{i \in I}$;
        \item Fractional matching $(x_{ij})_{(i,j) \in E}$ that satisfies Equations \eqref{eqn:xi}, \eqref{eqn:xj}, \eqref{eqn:edge} and \eqref{eqn:yj};
        \item Boundary times $t_0,t_1$ such that $0\leq t_0\leq t_1\leq 1$.
    \end{itemize}
    \smallskip
    \emph{When a first-class online vertex of type $i \in I_1$ arrives at time $0 \le t \le 1$:}
    \begin{itemize}[itemsep=0pt, topsep=4pt]
        \item Match it to the only neighbor if the neighbor is unmatched.
    \end{itemize}
    \smallskip
    \emph{When a second-class online vertex of type $i \in I_2$ arrives at time $0 \le t \le t_0$:}
    \begin{itemize}[itemsep=0pt, topsep=4pt]
        \item Discard it.
    \end{itemize}
    \smallskip
    \emph{When a second-class online vertex of type $i \in I_2$ arrives at time $t_0 < t \le t_1$:}
    \begin{itemize}[itemsep=0pt, topsep=4pt]
        \item Match it to each unmatched neighbor with probability $\frac12$.
    \end{itemize}
    \smallskip
    \emph{When a second-class online vertex of type $i \in I_2$ arrives at time $t_1 < t \le 1$:}
    \begin{itemize}[itemsep=0pt, topsep=4pt]
        \item If only one neighbor was unmatched at time $t_1$, match it to the neighbor if the neighbor is still unmatched;
        \item Otherwise match it to each unmatched neighbor with probability $\frac12$.
    \end{itemize}
\end{tcolorbox}

\subsection{Analysis}

We can see from the algorithm that the strategy of each online vertex is almost fixed. Only second-class vertices will adjust the strategy in the final stage based on a one-time observation at time $t_1$. This means the match events of different offline vertices are highly independent. We formalize this by first defining for each edge $(i,j)$ the \emph{match rate} $r_{ij}(t)$ to be the rate of online vertex $i$ trying to match offline vertex $j$ at time $t$. Note that for any first-class edge $(i,j)$, or for any second-class edge $(i,j)$ before time $t_1$, $r_{ij}(t)$ is independent of any randomness before. By the definition of the algorithm, we have:

\begin{fact}
\label{fact:rate}
	For any first-class edge $(i,j)$, for any time $0\le t\le 1$, $r_{ij}(t)=\lambda_i=x_{ij}$.
	
	For any second-class edge $(i,j)$:
	
	(1) For any time $0\le t\le t_0$, $r_{ij}(t)=0$;
	
	(2) For any time $t_0< t\le t_1$, $r_{ij}(t)=\frac12\lambda_i=x_{ij}$.
\end{fact}

For any second-class edge $(i,j)$ after time $t_1$, $r_{ij}(t)$ is only dependent on the status of the other neighbor of $i$ at time $t_1$. For simplicity, for a second-class edge $(i,j)$ where $j$ has been matched at time $t_1$, when online vertex $i$ arrives after time $t_1$, we artificially suppose it tries to match $j$ (and of course fails) with probability $\frac12$ (respectively $1$) if the other neighbor of $i$ is unmatched (respectively has been matched) at time $t_1$. Then we have:

\begin{fact}
\label{fact:rate2}
	For any second-class edge $(i,j)$, for any time $t_1< t\le 1$, letting $j'$ be the other neighbor of $i$:
	
	(1) If $j'$ is unmatched at time $t_1$, $r_{ij}(t)=\frac12\lambda_i=x_{ij}$;
	
	(2) If $j'$ has been matched at time $t_1$, $r_{ij}(t)=\lambda_i=2x_{ij}$.
\end{fact}

Using the match rate, we can bound the unmatched probability of each offline vertex at any time. For vertex-weighted matching, an upper bound of the unmatched probability of each offline vertex will be useful since it can directly contribute to the total weight. However, for edge-weighted matching, we need to compute the matched probability for each edge, so instead we lowerbound the unmatched probability of each offline vertex and then make use of the exact match rate of each edge.

For any offline vertex $j$, let $r_j(t)=\sum_{i\in I}r_{ij}(t)$ be the match rate of $j$ and $A_j(t)$ be the indicator of whether $j$ is unmatched at time $t$. By definition:

\begin{fact}
\label{fact:ratio}
	For any edge $(i,j)$ and any time $0\le t\le 1$, the probability that it has been matched at time $t$ is $\int_0^t \E[r_{ij}(t')A_j(t')] \mathrm{d}t'$.
\end{fact}

\begin{fact}
\label{fact:prob}
	For any offline vertex $j$ and any time $0\le t\le 1$, $\E[A_j(t)]=1-\int_0^t \E[r_j(t')A_j(t')] \mathrm{d}t'$.
\end{fact}

Combining Facts \ref{fact:rate}, \ref{fact:rate2} and \ref{fact:prob} we can get the following lemma. Recall that $y_j=\sum_{i\in I_1}x_{ij}$ is the total flow of first-class edges incident to offline vertex $j$.

\begin{lemma}
\label{lem:prob}
	For any offline vertex $j$:
	
	(1) For any time $0\le t\le t_0$, $\E [A_j(t)]=e^{-y_jt}$;
	
	(2) For any time $t_0< t\le t_1$, $\E [A_j(t)]=e^{-y_jt_0-(t-t_0)}$;
	
	(3) For any time $t_1< t\le 1$, $\E [A_j(t)]\ge e^{-y_jt_0-(t_1-t_0)-(2-y_j)(t-t_1)}$.
\end{lemma}

\begin{proof}
    (1) For $0\le t\le t_0$, $$r_j(t)=\sum_{i\in I_1}x_{ij}(t)+\sum_{i\in I_2}0=y_j$$ independent of $A_j(t)$, so $$\E[A_j(t)]=1-y_j\int_0^t \E[A_j(t')] \mathrm{d}t'=e^{-y_jt}.$$
    
    (2) For $t_0< t\le t_1$, $$r_j(t)=\sum_{i\in I_1}x_{ij}(t)+\sum_{i\in I_2}x_{ij}(t)=1$$ independent of $A_j(t)$, so $$\E[A_j(t)]=\E[A_j(t_0)]-\int_{t_0}^t \E[A_j(t')] \mathrm{d}t'=e^{-y_jt_0-(t-t_0)}.$$
    
    (3) For $t_1< t\le 1$, $$r_j(t)\le\sum_{i\in I_1}x_{ij}(t)+\sum_{i\in I_2}2x_{ij}(t)=y_j+2(1-y_j)=2-y_j$$ for any $A_j(t)$, so $$\E[A_j(t)]\ge\E[A_j(t_1)]-(2-y_j)\int_{t_1}^t \E[A_j(t')] \mathrm{d}t'.$$ Let $$f(t)=\E[A_j(t_1)]-(2-y_j)\int_{t_1}^t f(t') \mathrm{d}t',$$ then $\E[A_j(t_1)]=f(t_1)$ and $$\frac{\mathrm{d}\E[A_j(t)]}{\mathrm{d}t} \Big/ \E[A_j(t)]\ge 2-y_j = \frac{\mathrm{d}f(t)}{\mathrm{d}t} \Big/ f(t)$$ for $t_1<t\le 1$, so $$\E[A_j(t)]\ge f(t)=e^{-y_jt_0-(t_1-t_0)-(2-y_j)(t-t_1)}$$ for $t_1< t\le 1$.
\end{proof}

The match rate of a second-class edge after time $t_1$ is dependent on the status of another offline vertex at time $t_1$, so we also need the following lemma to bound the conditional matched probability.

\begin{lemma}
\label{lem:cond-prob}
	For any offline vertices $j\neq j'$, any $k\in\{0,1\}$ and any time $t_1< t\le 1$, \\ $\E[A_j(t)|A_{j'}(t_1)=k]\ge e^{-y_jt_0-(t_1-t_0)-(2-y_j)(t-t_1)}$.
\end{lemma}

\begin{proof}
    The argument is the same as that of Lemma \ref{lem:prob}. Since for $0\le t\le t_1$, $r_j(t)$ is also independent of $A_{j'}(t_1)$, $$\E[A_j(t_1)|A_{j'}(t_1)=k]=\E[A_j(t_1)]=e^{-y_jt_0-(t_1-t_0)}.$$ For $t_1< t\le 1$, given $A_{j'}(t_1)=k$, we still have $r_j(t)\le 2-y_j$, so $$\E[A_j(t)|A_{j'}(t_1)=k]\ge\E[A_j(t_1)|A_{j'}(t_1)=k]-(2-y_j)\int_{t_1}^t \E[A_j(t')|A_{j'}(t_1)=k] \mathrm{d}t',$$ so $$\E[A_j(t)|A_{j'}(t_1)=k]\ge e^{-y_jt_0-(t_1-t_0)-(2-y_j)(t-t_1)}.$$
\end{proof}

Putting things together, we can compute the matched probability for each edge to prove our main theorem:

\begin{theorem}
\label{thm:iid}
	Multistage Suggested Matching is $0.645$-competitive for edge-weighted online stochastic matching.
\end{theorem}

\begin{proof}
    We define the competitive ratio on each edge $(i,j)$ to be the ratio of its final matched probability to $x_{ij}$. Since the offline optimum can be bounded by $\sum_{(i,j)\in E}w_{ij}x_{ij}$, we only need to prove that the competitive ratio on every edge is at least $0.645$. 

	For any first-class edge $(i,j)$, by Facts \ref{fact:ratio} and \ref{fact:rate}, the matched probability of this edge is
    $$\int_{0}^{1}x_{ij}\E [A_j(t)] \mathrm{d}t.$$
    
    So the competitive ratio on this edge is:
	\begin{align*}
		\int_{0}^{1}\E [A_j(t)] \mathrm{d}t & = \int_{0}^{t_0}\E [A_j(t)] \mathrm{d}t + \int_{t_0}^{t_1}\E [A_j(t)] \mathrm{d}t + \int_{t_1}^{1}\E [A_j(t)] \mathrm{d}t \\
		& \ge \int_{0}^{t_0} e^{-y_jt} \mathrm{d}t + \int_{t_0}^{t_1}e^{-y_jt_0-(t-t_0)} \mathrm{d}t + \int_{t_1}^{1}e^{-y_jt_0-(t_1-t_0)-(2-y_j)(t-t_1)} \mathrm{d}t \tag{Lemma \ref{lem:prob}} \\
		& = \frac{1}{y_j}\big(1-e^{-y_jt_0}\big)+e^{-y_jt_0}\big(1-e^{-(t_1-t_0)}\big)+\frac{1}{2-y_j}e^{-y_jt_0-(t_1-t_0)}\big(1-e^{-(2-y_j)(1-t_1)}\big).
	\end{align*}

	Now consider any second-class edge $(i,j)$. Let $j'$ be the other neighbor of the online vertex type $i$. By Fact \ref{fact:ratio}, the matched probability of this edge is:
	$$\int_0^1 \E[r_{ij}(t)A_j(t)] \mathrm{d}t = \int_0^{t_0}\E[r_{ij}(t)A_j(t)] \mathrm{d}t + \int_{t_0}^{t_1} \E[r_{ij}(t)A_j(t)] \mathrm{d}t + \int_{t_1}^1 \E[r_{ij}(t)A_j(t)] \mathrm{d}t.$$
	
	By Fact \ref{fact:rate}, we have:
	$$\int_0^{t_0}\E[r_{ij}(t)A_j(t)] \mathrm{d}t=0$$ and $$\int_{t_0}^{t_1} \E[r_{ij}(t)A_j(t)] \mathrm{d}t=\int_{t_0}^{t_1} x_{ij}\E[A_j(t)] \mathrm{d}t.$$
	
	By Fact \ref{fact:rate2}, we have:
	$$\int_{t_1}^1 \E[r_{ij}(t)A_j(t)] \mathrm{d}t = \E[A_{j'}(t_1)]\int_{t_1}^1 x_{ij}\E[A_j(t)|A_{j'}(t_1)=1] \mathrm{d}t + (1-\E[A_{j'}(t_1)])\int_{t_1}^1 2x_{ij}\E[A_j(t)|A_{j'}(t_1)=0] \mathrm{d}t.$$
	
	So the competitive ratio on this edge is:
	\begin{align*}
		& \int_{t_0}^{t_1}\E [A_j(t)] \mathrm{d}t + \E [A_{j'}(t_1)] \int_{t_1}^{1}\E[A_j(t)|A_{j'}(t_1)=1] \mathrm{d}t + 2(1-\E [A_{j'}(t_1)]) \int_{t_1}^{1} \E[A_j(t)|A_{j'}(t_1)=0] \mathrm{d}t \\
		\ge & \int_{t_0}^{t_1}e^{-y_jt_0-(t-t_0)} \mathrm{d}t + \big(2-e^{-y_{j'}t_0-(t_1-t_0)}\big) \int_{t_1}^{1}e^{-y_jt_0-(t_1-t_0)-(2-y_j)(t-t_1)} \mathrm{d}t \tag{Lemmas \ref{lem:prob} and \ref{lem:cond-prob}} \\
		\ge & \ e^{-y_jt_0}\big(1-e^{-(t_1-t_0)}\big) + \big(2-e^{-(t_1-t_0)}\big)\frac{1}{2-y_j}e^{-y_jt_0-(t_1-t_0)}\big(1-e^{-(2-y_j)(1-t_1)}\big).
	\end{align*}
    
	Taking $t_0=0.05$ and $t_1=0.75$\footnote{As the ratios are complicated, we enumerate the values of the variables to get a nearly optimal solution. Further adjustment can only improve the ratio on the 4th digit after the decimal point.}, under Eqn. \eqref{eqn:yj}, the minimum values of both ratios are achieved at $y_j=1-\ln2$. The proof is basic but tedious calculus, so we defer it to Appendix \ref{app:iid}. When $y_j=1-\ln2$, both of them are at least $0.645$.
\end{proof}

\section*{Acknowledgments}

We thank the anonymous reviewers for their helpful comments.

\bibliographystyle{plainnat}
\bibliography{ref}

\appendix

\section{Missing Proof in Theorem \ref{thm:iid}}
\label{app:iid}

Define:
\begin{align*}
	a(x) & = \int_{0}^{t_0} e^{-tx} \mathrm{d}t = \frac{1}{x}\big(1-e^{-\frac{x}{20}}\big) , \\
	b(x) & = \int_{t_0}^{t_1}e^{-t_0x-(t-t_0)} \mathrm{d}t = e^{-\frac{x}{20}}\big(1-e^{-\frac{7}{10}}\big) , \\
	c(x) & = \int_{t_1}^{1}e^{-t_0x-(t_1-t_0)-(2-x)(t-t_1)} \mathrm{d}t = \frac{1}{2-x}e^{-\frac{x}{20}-\frac{7}{10}}\big(1-e^{-\frac{2-x}{4}}\big) .
\end{align*}
Then we only need to prove that $a(x)+b(x)+c(x)$ and $b(x)+(2-e^{-\frac{7}{10}})c(x)$ are decreasing in $x\in[0,1-\ln2]$. Apparently, $a(x)$ and $b(x)$ are decreasing, so it suffices to prove:
\[
    -b'(x) > \big(2-e^{-\frac{7}{10}}\big)c'(x),
\]
i.e.:
\[
    \frac{1}{20}e^{-\frac{x}{20}}\big(1-e^{-\frac{7}{10}}\big) > \big(2-e^{-\frac{7}{10}}\big)\frac{1}{2-x}e^{-\frac{x}{20}-\frac{7}{10}}\Big(\big(1-e^{-\frac{2-x}{4}}\big)\big(\frac{1}{2-x}-\frac{1}{20}\big)-\frac{1}{4}e^{-\frac{2-x}{4}}\Big).
\]
After rearranging terms, we get:
\[
    (2-x)\big(e^{\frac{7}{10}}-1\big) > 20\big(2-e^{-\frac{7}{10}}\big)\Big(\big(1-e^{-\frac{2-x}{4}}\big)\big(\frac{1}{2-x}-\frac{1}{20}\big)-\frac{1}{4}e^{-\frac{2-x}{4}}\Big).
\]
The left-hand-side achieves minimum value $1.716\dots$ at $x=1-\ln2$.
When $x\in[0,\frac{1-\ln2}{2}]$, the right-hand-side is no more than:
\[
    20\big(2-e^{-\frac{7}{10}}\big)\Big(\big(1-e^{-\frac{1}{2}}\big)\big(\frac{1}{2-(1-\ln2)/2}-\frac{1}{20}\big)-\frac{1}{4}e^{-\frac{1}{2}}\Big) = 1.256\dots
\]
When $x\in[\frac{1-\ln2}{2},1-\ln2]$, the right-hand-side is no more than:
\[
    20\big(2-e^{-\frac{7}{10}}\big)\Big(\big(1-e^{-\frac{2-(1-\ln2)/2}{4}}\big)\big(\frac{1}{1+\ln2}-\frac{1}{20}\big)-\frac{1}{4}e^{-\frac{2-(1-\ln2)/2}{4}}\Big) = 1.272\dots
\]
So the inequality follows.

\end{document}